\documentclass[11pt]{article}

\usepackage[T1]{fontenc}
\usepackage{amsmath,amsfonts,amssymb}
\usepackage[amsmath,amsthm,thmmarks]{ntheorem}
\usepackage{xspace}
\usepackage{xcolor}
\usepackage{algorithmic}
\usepackage{bm}
\usepackage{booktabs}
\usepackage{times}
\usepackage{fullpage}
\usepackage[top=0.88in,bottom=0.88in,left=0.88in,right=0.88in]{geometry}
\usepackage{nicefrac}
\usepackage{textcomp}
\usepackage{stmaryrd}
\usepackage{mathrsfs}
\usepackage{paralist}

\usepackage{enumitem}
\setlist[itemize]{leftmargin=*}
\setlist[enumerate]{leftmargin=*}
\usepackage{comment}
\usepackage{tikz}
\usepackage[font=footnotesize,labelsep=space,labelfont=bf]{caption}

\usepackage{url}

\usepackage[pdftex,bookmarks=true,pdfstartview=FitH,colorlinks,linkcolor=blue,filecolor=blue,citecolor=blue,urlcolor=blue,pagebackref=true]{hyperref}
\newtheorem{proposition}{Proposition}
\newtheorem{theorem}{Theorem}
\newtheorem{definition}{Definition}
\newtheorem{lemma}{Lemma}
\newtheorem{claim}{Claim}
\newtheorem{corol}[theorem]{Corollary}

\newcommand{\namedref}[2]{\hyperref[#2]{#1~\ref*{#2}}}

\newcommand{\Sectionref}[1]{\namedref{Section}{sec:#1}}

\newcommand{\Appendixref}[1]{\namedref{Appendix}{app:#1}}
\newcommand{\Theoremref}[1]{\namedref{Theorem}{thm:#1}}
\newcommand{\Corollaryref}[1]{\namedref{Corollary}{cor:#1}}
\newcommand{\Propositionref}[1]{\namedref{Proposition}{prop:#1}}

\newcommand{\Lemmaref}[1]{\namedref{Lemma}{lem:#1}}
\newcommand{\Claimref}[1]{\namedref{Claim}{clm:#1}}

\newcommand{\Footnoteref}[1]{\namedref{Footnote}{foot:#1}}
\newcommand{\Pageref}[1]{\hyperref[#1]{page~\pageref*{#1}}}

\newcommand{\eps}{\ensuremath{\epsilon}\xspace}
\newcommand{\Real}{\ensuremath{{\mathbb R}}\xspace}
\DeclareMathOperator{\SDiff}{SD}

\definecolor{darkred}{rgb}{0.5, 0, 0} 
\definecolor{darkblue}{rgb}{0,0,0.5} 
\hypersetup{
    colorlinks=true,
    linkcolor=darkred,
    citecolor=darkblue,
    urlcolor=darkblue   
}

\newcommand{\pr}{\operatorname{Pr}}
\newcommand{\prob}{\mathbf{p}}

\newcommand{\outputs}{\ensuremath{\mapsto}\xspace}

\newcommand{\Pirate}[2][]{\ensuremath{{CC}^{\scriptscriptstyle #1}_{\scriptscriptstyle #2}(\Pi)}\xspace}
\newcommand{\Piratex}{\ensuremath{{CC}(\Pi)}\xspace}
\newcommand{\CC}[3][]{\ensuremath{CC^{#1}_{\scriptscriptstyle #2}\left(#3\right)}\xspace}

\newcommand{\PiIC}[2][]{\ensuremath{{IC}^{\scriptscriptstyle #1}_{\scriptscriptstyle #2}(\Pi)}\xspace}
\newcommand{\IC}[2]{\ensuremath{IC_{\scriptscriptstyle #1}\left(#2\right)}\xspace}

\newcommand{\Outerreg}{\ensuremath{\mathfrak{R}}\xspace}
\newcommand{\ICostreg}{\ensuremath{\mathfrak{I}}\xspace}
\newcommand{\Capacityreg}{\ensuremath{\mathfrak{C}}\xspace}

\newcommand{\Outer}[2]{\ensuremath{\Outerreg\left(#2:#1\right)}\xspace}
\newcommand{\Outeradj}[2]{\ensuremath{\widetilde{\Outerreg}\left(#2:#1\right)}\xspace}
\newcommand{\ICost}[2]{\ensuremath{\ICostreg\left(#2:#1\right)}\xspace}
\newcommand{\Capacity}[3][]{\ensuremath{{\Capacityreg}_{#1}\left(#3:#2\right)}\xspace}

\newcommand{\RT}{\ensuremath{\mathfrak{T}}\xspace}
\newcommand{\Rtens}[2]{\ensuremath{\RT({#1};{#2})}\xspace}
\newcommand{\CIwyn}{\ensuremath{CI_\text{\sf Wyn}}\xspace}
\newcommand{\CIgk}{\ensuremath{CI_\text{GK}}\xspace}
\newcommand{\Twyn}{\ensuremath{T_\text{\sf Wyn}}\xspace}
\newcommand{\Tgk}{\ensuremath{T_\text{GK}}\xspace}

\newcommand{\Discrepancy}[2]{\ensuremath{\mathrm{Disc}_{\scriptscriptstyle #2}({#1})}\xspace}

\renewcommand{\paragraph}[1]{\smallskip\noindent{\bf #1}~}

\newcommand{\disc}{\ensuremath{\Delta}\xspace}

\newcommand{\cA}{\ensuremath{\mathcal{A}}\xspace}
\newcommand{\cB}{\ensuremath{\mathcal{B}}\xspace}

\newcommand{\cL}{\ensuremath{\mathcal{L}}\xspace}
\newcommand{\cM}{\ensuremath{\mathcal{M}}\xspace}
\newcommand{\cQ}{\ensuremath{\mathcal{Q}}\xspace}
\newcommand{\cR}{\ensuremath{\mathcal{R}}\xspace}
\newcommand{\cS}{\ensuremath{\mathcal{S}}\xspace}
\newcommand{\cT}{\ensuremath{\mathcal{T}}\xspace}
\newcommand{\cX}{\ensuremath{\mathcal{X}}\xspace}
\newcommand{\cY}{\ensuremath{\mathcal{Y}}\xspace}
\newcommand{\Adv}[1]{\ensuremath{D(#1)}\xspace}
\newcommand{\RR}{\ensuremath{\hat{R}}\xspace}
\newcommand{\QQ}{\ensuremath{Q'}\xspace}

\begin{document}
\title{Tension Bounds for Information Complexity}
\author{Manoj M. Prabhakaran%
        \thanks{%
			Department of Computer Science,
					   University of Illinois, Urbana-Champaign.
					   {\tt mmp@illinois.edu}.}
		\and  
		Vinod M. Prabhakaran%
		\thanks{%
			School of Technology and Computer Science,
				   Tata Institute of Fundamental Research, Mumbai, India.
				   {\tt vinodmp@tifr.res.in}}
}

\maketitle

\begin{abstract}
The main contribution of this work is to relate information complexity to
``tension'' \cite{PrabhakaranPr14} -- an information-theoretic quantity
defined with no reference to protocols -- and to illustrate that it allows
deriving strong lower-bounds on information complexity.  In particular, we
use a very special case of this connection to give a quantitatively tighter
connection between information complexity and discrepancy than the one in
\cite{BravermanWe12} (albeit, restricted to independent inputs). Further, as
tension is in fact a multi-dimensional notion, it enables us to bound the
2-dimensional region that represents the {\em trade-off} between the amounts
of communication in the two directions, in a 2-party protocol.

This work is also intended to highlight tension as a fundamental measure of
correlation between a pair of random variables, with rich connections to a
variety of questions in computer science and information theory.
\end{abstract}

\thispagestyle{empty}
\newpage

\setcounter{page}{1}

\pagestyle{plain}

\section{Introduction}
\label{sec:intro}

Communication complexity, since the seminal work of Yao~\cite{Yao79}, has
been a central question in theoretical computer science. Many of the recent
advances in this area have centred around the notion of information
complexity, which measures the {\em amount of information} about the inputs
-- rather than the {\em number of bits} -- that should be present in a
protocol's transcript, if it should compute a function (somewhat) correctly.

The main contribution of this work is to relate information complexity to
``tension'' \cite{PrabhakaranPr14} -- an information-theoretic quantity
defined with no reference to protocols -- and to illustrate that it allows
deriving strong bounds on information complexity.  In particular, we use a
very special case of this connection to give a quantitatively tighter
connection between information complexity and discrepancy than the one in
\cite{BravermanWe12} (albeit, restricted to independent inputs). Further, as
tension is in fact a multi-dimensional notion, it enables us to bound the
2-dimensional region that represents the {\em trade-off} between the amounts
of communication in the two directions, in a 2-party protocol.

This work is also intended to highlight tension as a fundamental measure of
correlation between a pair of random variables, with rich connections to a
variety of questions in computer science and information theory. Tension is
intimately related to the notion of {\em common information} developed in
highly influential works in the information theory literature from the
70's \cite{GacsKo73,Wyner75}. Tension has proven useful in deriving
state-of-the-art bounds on ``cryptographic complexity'' (i.e., number of
instances of, say, oblivious transfer needed per instance of securely
computing a function) \cite{PrabhakaranPr14} and communication complexity of
information-theoretically secure multiparty computation \cite{DataPrPr14}.
However, currently we have few tools to compute (or bound) tension. We leave
it as an important problem to understand tension in general as well as for
specific random variables.

\paragraph{What is Tension?} Tension of a pair of correlated random
variables $(A;B)$ captures ``non-trivial'' correlation between them: i.e.,
the extent to which correlation {\em cannot} be captured by a common random
variable that can be associated with both $A$ and $B$.  The question of how
well correlation {\em can } be captured by a random variable is formulated
in terms of ``common information.'' Two different notions of common
information were developed in the 70's, $\CIgk(A;B)$ by G\'acs-K\"orner
\cite{GacsKo73}, and $\CIwyn(A;B)$ by Wyner \cite{Wyner75}, with operational
meanings related to certain natural information theoretic problems. (See \Appendixref{tension}
for more details.) One can define
corresponding notions of tension as the gap between mutual information
(which accounts for all the correlation, but may not correspond to a common
random variable) and common information.  More precisely, one can define the
non-negative tension quantities $\Tgk(A;B)=I(A;B)-\CIgk(A;B)$ and
$\Twyn(A;B)=\CIwyn(A;B)-I(A;B)$.  These notions of tension were identified in
\cite{PrabhakaranPr14} as special cases of a unified 3-dimensional notion of {\em
tension region}. 

In \cite{PrabhakaranPr14}, an operational meaning was attached
to tension region in terms of a communication problem, and also it was shown
that a {\em secure} 2-party protocol for sampling correlated random
variables with ``high tension''%
\footnote{Informally, the farther the tension region is from the origin, the
higher the tension, along different dimensions.}
will need a large number of instances of oblivious transfer.  In
\Appendixref{tension}, we summarize some of the basic properties of the
tension region, as developed in \cite{PrabhakaranPr14}. 

We lower bound the information complexity of a function $f$ in terms of how
different the tension regions of $(X;Y)$ and $(X,Z; Y,Z)$ are, where
$Z=f(X,Y)$ (or rather, $\pr[Z=f(X,Y)]\ge\frac12+\eps$).  In particular, when
the inputs $(X;Y)$ are independent of each other (so that their tension is
zero, and hence contains the origin), the {\em information complexity
region} is shown to lie inside the tension region of $(X,Z;Y,Z)$. (An
information complexity region farther from the origin corresponds to a
higher lower-bound on information complexity.)  Note that even though $Z$
may be a single bit, the difference between the tension regions of $(X;Y)$
and $(XZ;YZ)$ could be quite large -- as we illustrate by the connection
with discrepancy.

\subsection{Overview of Results and Techniques}

Our contributions are in two parts: 
\begin{enumerate}
\item We show that information complexity can be lower-bounded using tension
-- a fundamental quantity defined with no reference to protocols.
\item We illustrate the potential of this approach for yielding strong
lower-bounds, by obtaining an improved lower-bound on information complexity
in terms of discrepancy.  
\end{enumerate}
Below, we shall elaborate on these contributions further. We point out that
our model and results are, in some ways, more general than prior work:
\begin{itemize}
\item In developing the connection between information complexity and
tension (as well as between information complexity and communication
complexity), we work with a ``bigger picture'' that considers 2-dimensional
notions of these quantities. We remark that even if we are interested only in bounding
communication complexity and information complexity (corresponding to
1-dimensional regions), using bounds in terms of the 2-dimensional region
can yield potentially stronger lower-bounds.
\item Our results hold for randomized functions, with asymmetric outputs.
\item A minor difference is that in our communication model, we allow for
the possibility that the transcript (i.e., the concatenation of all the
messages sent during the protocol in either direction) may not be
``parsable'' into individual messages by an outsider, though each party,
with its input can parse it. (See \Footnoteref{model}.)
\end{itemize}

We propose, as a direction for further study, that various results on
information complexity which led to advances in communication complexity can
be rederived for tension, thereby providing alternate (and hopefully
simpler) proofs to these results.  Also, we leave it as an open problem to
exploit the full power of the tension bounds: currently, there are few
techniques to map out the full 3-dimensional tension region of a pair of
random variables.

\subsubsection*{Tension, Information Complexity and Communication Complexity}
The basic idea behind lower-bounding information complexity by tension is,
in fact, easy to see. Consider a protocol in which, for simplicity, the two
parties are given independent inputs $X,Y$, exchange messages to generate a
transcript $M$, and produces a common output $Z$. Since $X,Y$ were
independent of each other, we know that $(X,Z)$ and $(Y,Z)$ should continue
to be independent conditioned on the transcript, $M$; i.e., $(X,Z)-M-(Y,Z)$.
One can see that the information cost of this protocol $I(X;M|Y)+I(Y;M|X)$
can be lower bounded by $I(XZ;M|YZ)+I(YZ;M|XZ)$, which in turn can be lower
bounded by $\inf_{Q:XZ-Q-YZ} I(XZ;Q|YZ)+I(YZ;Q|XZ)$ (i.e., without requiring
that $Q$ is the transcript of a protocol that outputs $Z$, but only that
$XZ-Q-YZ$).  The latter quantity is exactly the Wyner-Tension,
$\Twyn(XZ;YZ)$.  When $(X,Y)$ are not independent, this lower-bound changes
to $\Twyn(XZ;YZ)-\Twyn(X;Y)$.  Jumping ahead, we mention that we can extend
this basic lower-bound to a more general one, where we also consider $Q$
such that the condition $XZ-Q-YZ$ is replaced by $I(XZ;YZ|Q)\le c$ for $c\ge
0$ (this is of interest only when $X,Y$ are correlated). 

We derive our lower-bounds in terms of 2-dimensional regions, which can
potentially yield stronger lower bounds than considering the two points
$\Twyn(XZ;YZ)$ and $\Twyn(X;Y)$ on the one-dimensional line.
The general relation between communication complexity and information
complexity, and that between information complexity and tension
(\Theoremref{ic-cc} and \Theoremref{tens-ic}) can be summarized as
\[ \Capacityreg \subseteq \ICostreg \subseteq \Outerreg, \]
where \Capacityreg denotes the set of communication cost pairs (number of
bits from Alice to Bob, and vice-versa) achievable by protocols computing a
possibly randomized function $f$, \ICostreg denotes the information cost pairs (information
communicated by Alice to Bob about her input, and vice versa) achievable by
such protocols, and \Outerreg, as described below, denotes a 2-dimensional restriction of the
3-dimensional ``tension region''  that was introduced in
\cite{PrabhakaranPr14}. Here, all three regions are defined to be ``upward closed''
subsets of $\Real_+^2$: i.e., if $(x,y)$ is in the set and then so is
$(x',y')$ for all $x'\ge x$ and $y'\ge y$.

Before fully describing \Outerreg, for simplicity, consider  the case of
independent $X,Y$. In this case, \Outerreg is given by
\[ \RT_0(XZ;YZ) = \{ (r_1,r_2)\in\Real_+^2: \exists Q  \text{ s.t. } 
	XZ-Q-YZ \text{ and } I(XZ;Q|YZ) \le r_1,\;I(YZ;Q|XZ) \le r_2 \}. \]
This is a convex, upward-closed region, typically bounded away from the
origin. In the more general case, when $X,Y$ are not independent,
\Outerreg is somewhat more complex. In particular, it is contained in the
region
\[ \RT_0(XZ;YZ) - \RT_0(X;Y) = \{ (r_1,r_2)\in\Real_+^2: (r_1,r_2) +
\RT_0(X;Y) \subseteq \RT_0(XZ;YZ) \}. \]
Typically, we expect the region $\RT_0(XZ;YZ)$ to be much further away from
the origin than $\RT_0(X;Y)$ (i.e., $(XZ;YZ)$ has much higher tension than
$(X;Y)$).  The region $\RT_0(XZ;YZ) - \RT_0(X;Y)$  (or rather, the lower
boundary of it) captures the least amount by which $\RT_0(X;Y)$ should be
pushed away from the origin so that it moves completely inside
$\RT_0(XZ;YZ)$.
The bound $\Twyn(XZ;YZ)-\Twyn(X;Y)$ mentioned
earlier, can be obtained as 
\[
\inf_{(a,b)\in\RT_0(XZ;YZ)}(a+b) - \inf_{(a,b)\in\RT_0(X;Y)}(a+b)
\le \inf_{(a,b)\in \RT_0(XZ;YZ) - \RT_0(X;Y)} (a+b).
\]
Here we point out that the inequality above could be strict, in which case
settling for a 1-dimensional version would give a weaker bound than what is
implied by the 2-dimensional version.

The full definition of $\Outerreg$ is $\cap_{c\ge 0} \RT_c(XZ;YZ) -
\RT_c(X;Y)$, where in $\RT_c(XZ;YZ)$ we do not restrict to $Q$ such that
$XZ-Q-YZ$; instead we require only that $I(XZ;YZ|Q)\le c$. In showing that
\Outerreg gives a valid outer-bound on \ICostreg, we rely on a certain
``monotonicity'' property of the 3-dimensional tension region of the views
of the parties in a protocol: the tension region can only extend closer to
the origin as the protocol progresses.%
\footnote{A more general monotonicity property holds, allowing the parties
to not just exchange messages, but also to ``securely'' delete parts of
their views. This was shown in \cite{PrabhakaranPr14} for all of the tension
region, including \Twyn; a similar result appeared for \Tgk and two other
points in the tension region in an earlier work of Wolf and Wullschleger
\cite{WolfWu05}.}

While quite general in its form, we leave it as an open problem to exploit
the full power of this connection, since understanding the full
3-dimensional tension region is an outstanding challenge.

\paragraph{Information Complexity vs.\ Communication Complexity.}
As mentioned above, the connection between information complexity and
communication complexity is well-known. We extend this relation to the
2-dimensional regions $\Capacityreg$ and $\ICostreg$. Note that \Capacityreg
corresponds to {\em average} communication-complexity. Hence $\Capacityreg
\subseteq \ICostreg$ directly yields a lower bounds not just on worst-case
communication complexity (as it is often presented in the literature), but
in fact on average communication complexity as well.%
\footnote{In fact, we observe that the inequality $\PiIC\mu \le \Piratex$
\cite{BravermanRa11} used to relate information cost and worst-case
communication cost of a protocol can in fact be strengthened to $\PiIC\mu
\le \Pirate\mu \le \Piratex$, for any distribution $\mu$ over the inputs.
(See \Lemmaref{ic-cc}.)}
This allows one to translate lower-bounds on information complexity of
protocols of a certain error rate to lower-bounds on average communication
complexity for the same error rate.

\subsection*{Discrepancy vs.\ Tension}
Consider $X,Y$ being $n$-bit long strings, and $Z$ being a single bit with
$\Pr[Z=f(X,Y)]\ge\frac12+\eps$, where $f$ is, say, the inner-product
over $GF(2)$. When $X,Y$ are independent, $\Twyn(X;Y)=0$.
One would wonder if adding a single bit to the random variables can
change their tension by more than a constant amount. But as it turns out, 
the correlation between $XZ,YZ$ as captured by \Twyn can be $\Omega(n)$
bits! For this, we rely on the function $f$ having an exponentially small ``discrepancy,''
a combinatorial measure of complexity of a function.

Indeed, in \Sectionref{disc-tens} we show that
the Wyner-Tension $\Twyn(XZ;YZ)$, where $X,Y$ are independent, and
$\Pr[Z=f(X,Y)]\ge\frac12+\eps$, can be lower-bounded as
$\Omega(\eps\log\frac\eps\disc)$ if the discrepancy of $f$ (w.r.t.\ the
distribution of $(X,Y)$) is upper-bounded by $\disc$. This compares
favorably with a similar bound in \cite{BravermanWe12}, of the form
$\Omega(\eps^2\log\frac\eps\disc)$ (though, as mentioned above, the bound in
\cite{BravermanWe12} applies even if $X,Y$ are not independent).

% This is perhaps the more technically involved part of this paper.
% However, we believe it is simpler and more direct (and quantitatively,
% tighter) than the one in \cite{BravermanWe12}.  

To lower-bound $\Twyn(XZ;YZ)$ it turns out to be enough to lower-bound
$I(XY;Q)$ such that $X-Q-Y$ and given $Q$, $Z$ is determined (i.e.,
$H(Z|Q)=0$). The high-level intuition is to analyze the advantage $Z$ has
(i.e., $\pr[Z=f(X,Y)]-\frac12$) as contributed by different values of $Q$.
For starters, suppose the input distribution is uniform and further, for
each value $q$ for $Q$, the conditional distribution $p_{XY|Q=q}$ is also
{\em uniform over a rectangle}.  Then, for $q$ such that this rectangle is
large, its contribution to the advantage will be small, because otherwise it
will result in a large discrepancy (recall that $Z$ must take a single value
conditioned on $Q=q$). Thus, to achieve a large advantage when the
discrepancy is small, most of the mass on $Q$ should correspond to $q$ such
that $p_{XY|Q=q}$ is uniform over a ``small'' rectangle. Intuitively, this
should imply a large value for $I(XY;Q)$.

This idea runs into several complications. Mainly, $p_{XY|Q=q}$ is
guaranteed only to be a product distribution, and not necessarily uniform
over its support. To tackle this, we show how to {\em slice} this distribution
into several components, each of which is indeed uniform (or more generally,
when $XY$ is not uniform, each one is $p_{XY|(X,Y)\in r}$ for some rectangle
$r$). One could then repeat the above argument with respect to the slices.
However, including the index of the slice into $Q$ would result in a large gap
between its mutual information with $XY$, and that of the original $Q$.
Instead we add a single bit to $Q$ to indicate whether the slice is a large
rectangle or a small rectangle. We then argue that collecting the small
rectangles into one single subset will still result in a (relatively) small
subset. With this, the above outline can indeed be made to work.

We remark that the intuition that if, for most $q$, the support of
$p_{XY|Q=q}$ has a small mass in the original distribution $p_{XY}$, then
$I(XY;Q)$ should be large is formalized in \Lemmaref{Ibound}. This may be of
independent interest.

\subsection{Related Work}

Many of the recent advances in the field of communication complexity \cite{Yao79}
have followed from using various notions of information complexity.
Earlier notions of information complexity appeared implicitly in several
works \cite{Ablayev96,PonzioRaVe01,SaksSu02}, and was first explicitly
defined in \cite{ChakrabartiShWiYa01}. The current notion of (internal)
information complexity originated in \cite{Bar-YossefJaKuSi04}. Information
complexity has been extensively used in 
in the recent communication complexity
literature
\cite{BravermanRa11,Braverman12,BravermanWe12,ChakrabartiKoWa12,KerenidisLaLeRoXi12,BarakBrChRa13}.
The notion was also adapted to specialized models or tasks
\cite{JayramKuSi03, JainRaSe03, JainRaSe05,HarshaJaMcRa10}. The result in
\cite{BravermanWe12} (since generalized by \cite{KerenidisLaLeRoXi12})
relates most to the result we derive to illustrate the potential of tension
bounds.

The notion of common information, to which tension is closely related, was
developed in the information-theory
literature~\cite{GacsKo73,Wyner75,AhlswedeKo74,PrabhakaranPr14}. Recently,
it has found use in communication complexity, cryptography and other
problems in theoretical computer science, 
e.g.~\cite{HarshaJaMcRa10,BraunPo13,BraunJaLePo13,DataPrPr14}. Some special
cases of tension were implicit in the work of Wolf and Wullschleger
\cite{WolfWu05}, who used their monotonicity properties in a protocol to
lower-bound the number of oblivious transfers needed for various secure
computation tasks. The full-fledged notion of tension region was developed
in \cite{PrabhakaranPr14}. A multi-party notion of tension was defined in
\cite{PrabhakaranPr12}.

\section{Preliminaries}
\label{sec:prelims}

\paragraph{Notation.}
For brevity of notation, we shall often
denote the random-variables
$(X,Y)$ etc.\ by $XY$ etc.
Also, we shall often use a random variable to denote the probability
distribution of the random variable, when the random variables that it
is jointly distributed with are clear from the
context: i.e., we may write $Q$ instead of $\prob_{Q|XY}$.
We write $A-Q-B$ to indicate that $I(A;B|Q)=0$.

\paragraph{Communication Complexity.}
Let $\Pi(X;Y)$ be a (randomized) 2-party protocol with inputs to the two
parties being $X$ and $Y$ respectively. The two parties alternate sending
messages to each other; $\Pi$ specifies which party sends the first message,
and the function mapping each party's current view to the distribution over
the next message that it sends, and a distribution over an optional output
it produces (on producing an output, the party halts).  The messages can be
of arbitrary length, but should be self-terminating given the transcript so
far, and either of the two inputs.%
\footnote{The traditional definition of a protocol in the communication
complexity literature is slightly more restrictive: it requires that the
messages are self-truncating, given just the transcript so far. We note that
when the two parties have correlated inputs (e.g., as part of their private
inputs, they share a one-time pad which is used to mask the entire
communication) this should no more be required.\label{foot:model}}
For simplicity, we do not include public coins in our model; however, with
suitable modifications in the definitions, all our results would continue to
hold in such a model. In particular, we note that tension between two random
variables is not altered by adding a common random variable (i.e., the
public random tape) to both the random variables.

We write $\Pi(X;Y)\outputs(A;B)$ to denote that the random
variables $(A;B)$ (jointly distributed with $(X;Y)$) are the outputs
produced by the two parties on running $\Pi(X;Y)$. We denote by
\Pirate[(12)]{XY} (respectively, \Pirate[(21)]{XY})
the expected
number of bits sent by party 1 to party 2 (respectively, by party 2 to party
1) in the protocol $\Pi(X;Y)$; the expectation is over the randomness of the
protocol, as well as the input distribution $\prob_{XY}$.

The communication complexity -- or more precisely, the ``achievable
communication rate region'' -- for computing $(A;B)$ given $(X;Y)$, is
defined as:
\[\Capacity{X;Y}{A;B} = \{ (r_1,r_2)  \in \Real_+^2 : \exists \Pi \text{ s.t. } \Pi(X;Y) \outputs (A;B) 
\text{ and } \Pirate[(12)]{XY} \le r_1, \Pirate[(21)]{XY} \le r_2 \}. \]
Note that the region \Capacity{X;Y}{A;B} is an {\em upward closed region}.
In fact, the different regions we shall define and use are all upward
closed. 

A special case of interest is when the $A=B=f(X,Y)$, for a boolean function
$f:\cX\times\cY \rightarrow \{0,1\}$. 
In this case we shall typically require of a protocol that the two parties agree on the
outcome, but we shall allow the outcome to be wrong with some probability
$\eps$ (probability taken over the input distribution as well as the
randomness of the protocol). We define the {\em communication complexity region}
for $f$ (for an error probability $\eps$) to be:
\[ \Capacity[\eps]{X;Y}{f} 
= \bigcup_{\substack{\prob_{Z|XY}:\\\SDiff(\prob_{ZXY},\prob_{f(X,Y)XY})\le\eps}} \Capacity{X;Y}{Z;Z},\]
where $\SDiff(p_A,p_B)$ is the total variation distance between the distributions $p_A,p_B$ defined as $\SDiff(p_A,p_B)=\frac{1}{2}\sum_a |p_A(a)-p_B(a)|$.
Also of special interest is the (average-case) {\em communication
complexity}, which considers just the total number of bits communicated,
irrespective of the direction:
\[ \CC[\eps]{XY}{f} = \inf\; \{ r_1+r_2 : (r_1,r_2) \in \Capacity[\eps]{X;Y}{f} \}. \]

\paragraph{Information Complexity.}
The {\em information cost} of a protocol $\Pi$ is defined as follows. Let $\Pi(X;Y)\outputs (A;B)$
and let $M$ denote the transcript of $\Pi(X;Y)$. Then we define
\begin{align*}
\PiIC[(12)]{XY} &= I(X;M|Y), & \PiIC[(21)]{XY} &= I(Y;M|X).
\end{align*}
Then, $\PiIC{XY} = \PiIC[(12)]{XY}+\PiIC[(21)]{XY}$. We define the {\em information complexity region} as:
\[
\ICost{X;Y}{A;B} =  \{ (r_1,r_2) \in \Real_+^2 : \exists \Pi \text{ s.t. }  \Pi(X;Y) \outputs (A;B) \text{ and } 
	 \PiIC[(12)]{XY} \le r_1, \PiIC[(21)]{XY} \le r_2 \}.
\]
Of special interest is the following quantity --- the information complexity of computing $Z$ from $(X;Y)$.
\[ \IC{XY}{Z} = \inf\; \{ r_1+r_2 : (r_1,r_2) \in \ICost{X;Y}{Z;Z} \}. \]

\paragraph{Discrepancy.}
Let
$\cR=\{ \cX'\times\cY' : \cX'\subseteq \cX, \cY' \subseteq \cY \}$), the set
of all ``rectangles'' in $\cX\times\cY$. Then, given a distribution $\prob_{XY}$ over
$\cX\times\cY$, and a boolean function $f:\cX\times\cY\rightarrow\{0,1\}$,
we define
\begin{align*}
\Discrepancy{f}{XY} 
&= \max_{r\in\cR}
		\left| \pr[(X,Y)\in r \land f(X,Y)=0] - \pr[(X,Y)\in r \land f(X,Y)=1] \right| \\
&= \max_{\substack{\cX'\subseteq \cX,\\\cY'\subseteq\cY}}
		\left| \sum_{\substack{(x,y)\in\cX'\times\cY':\\f(x,y)=0}} \prob_{XY}(x,y) - \sum_{\substack{(x,y)\in\cX'\times\cY':\\f(x,y)=1}}\prob_{XY}(x,y)\right|.
\end{align*}

\subsection{Tension}
\label{sec:prelims-tension}

The tension region of a pair of random variables was defined in~\cite{PrabhakaranPr14} as the following upward closed region.
\begin{definition}
For a pair of random variables $A,B$, their {\em tension region}  \Rtens AB is defined as 
\begin{align*}
\Rtens AB = \{ (r_1,r_2,r_3):&\; \exists Q \text{ jointly distributed with } A,B \\
          &\text{ s.t. } I(B;Q|A) \le r_1, I(A;Q|B) \le r_2, I(A;B|Q) \le r_3 \}.
\end{align*}
\end{definition}
As shown in~\cite{PrabhakaranPr14}, without loss of generality, we may assume a cardinality bound $|\cQ|\leq|\cA||\cB|+2$ on the alphabet $\cQ$ in the above definition, where $\cA$ and $\cB$ are the alphabets of $A$ and $B$, respectively. It was also shown there that \Rtens AB has the interpretation as a rate-information tradeoff region for a distributed common randomness generation problem which generalizes the common randomness problem of G\'acs and K\"orner~\cite{GacsKo73}. \Rtens AB is a closed, convex region, with the following monotonicty property for randomized (public/private coins) protocols: Suppose $X$,$Y$ are the inputs and $A$,$B$ the outputs of the parties under a protocol. Let $M$ denote the transcript of the protocol. Let $V_A=(X,A,M)$ and $V_B=(Y,B,M)$ denote the views of the parties at the end of the protocol.
\begin{proposition}[Theorem~5.4 of \cite{PrabhakaranPr14}]\label{prop:monotonicity}
$\Rtens {V_A}{V_B} \supseteq \Rtens XY$.
\end{proposition}

In the sequel we will apply certain implications of the above result. Specifically, we will be interested in the inclusion relationship of certain restrictions of the tension regions of inputs and the views. For convenience, we define for $c\ge0$ the intersection of tension region with the plane $r_3=c$ as $\RT_c$. More precisely,
\begin{align*} \RT_c(A;B)&= \{ (r_1,r_2)\in\Real_+^2: (r_1,r_2,c)\in \Rtens{A}{B}\} \\
&= \{ (r_1,r_2)\in\Real_+^2: \exists \prob_{Q|A,B} \text{ s.t. } 
	I(B;Q|A) \le r_1,\;I(A;Q|B) \le r_2,\; I(A;B|Q) \le c\}.
\end{align*}
The case of $c=0$ will be of special interest to us. Here, we will focus on the minimum $r_1+r_2$. We define the {\em Wyner-tension} $\Twyn(A;B)$ of two jointly distributed random variables $A,B$ as
\[ \Twyn(A;B) = \inf \{ r_1+r_2 : (r_1,r_2)\in \RT_0(A;B) \}
= \inf_{\substack{\prob_{Q|AB}:\\A-Q-B}} I(A;Q|B)+I(B;Q|A). \]
This quantity is related to Wyner's common information $\CIwyn(A;B)$ of two random variables $A,B$~\cite{Wyner75}.  
\[ \CIwyn(A;B) = \inf_{\substack{\prob_{Q|AB}:\\A-Q-B}} I(A,B;Q).\]
It is easy to see the following~\cite{PrabhakaranPr14}.
\[\Twyn(A;B) = \CIwyn(A;B) - I(A;B).\]
Notice that $\CIwyn(A;B)\geq I(A;B)$ and $\Twyn(A;B)\geq 0$.

\section{Tension vs.\ Information Complexity}
\label{sec:tens-ic}

In this section, we lower-bound information complexity in terms
of tension. As we shall work with the more general information complexity
region \ICost{X;Y}{A;B}, the ``lower-bound'' corresponds to bounding the region away from the
origin. For this, we shall define a region
$\Outer{X;Y}{A;B}\subseteq{\mathbb R}_+^2$, which will then be used o outer-bound
the region \ICost{X;Y}{A;B}. We define:
\begin{align*}
\Outer{X;Y}{A;B} = \bigcap_{c\geq0} \big( \RT_c(B,Y;A,X) - \RT_c(Y;X) \big),
\end{align*}
where $S_1-S_2 = \{ (a,b) \in \Real_+^2 : (a,b)+S_2 \subseteq S_1 \}$ and $(a,b)+S$, for $a,b\in\Real$ and $S\subseteq \Real^2$, is $\{(x,y)\in\Real^2:  (x+a,y+b) \in S\}$.
We also define
\begin{align*}
\Outeradj{X;Y}{A;B} = (H(B|Y)-H(AB|XY),H(A|X)-H(AB|XY)) + \Outer{X;Y}{A;B}.
\end{align*}
Note that if $H(A|X)\ge H(AB|XY)$ and 
$H(B|Y)\ge H(AB|XY)$, then $\Outeradj{X;Y}{A;B} \subseteq \Outer{X;Y}{A;B}$.
These conditions
are satisfied if, for instance, $A=B$ (both parties output the same value),
or $H(A,B|X,Y)=0$ (the output is a deterministic function of the input), or
more generally if $H(A|B,X,Y)=H(B|A,X,Y)=0$ (i.e., any randomness in the
outputs given the inputs is common to both outputs). Even if these
conditions are not satisfied, if the outputs $A$ and $B$ are short,
then $\Outeradj{X;Y}{A;B}$ is close to $\Outer{X;Y}{A;B}$, and the
difference between the two can be ignored.
\begin{theorem} \label{thm:tens-ic}
$\ICost{X;Y}{A;B} \subseteq \Outeradj{X;Y}{A;B}$.
In particular, if $H(A|X) \ge H(A,B|X,Y)$ and $H(B|Y)\ge H(A,B|X,Y)$, then,
\[\ICost{X;Y}{A;B} \subseteq \Outer{X;Y}{A;B}.\]
\end{theorem}

\begin{proof}
Consider any protocol $\Pi$ that takes $(X;Y)$ as input and outputs $(A;B)$.
Let $U_A=(X,A)$, $U_B=(Y,B)$, the input-output of Alice and Bob; and let $M$
be the transcript of the messages exchanged between Alice and Bob.
\begin{align}
\PiIC[(12)]{XY} &= I(X;M|Y) \\
&\stackrel{\text{(a)}}{=} I(X;M,B|Y)= I(X;B|Y) + I(X;M|Y,B) \notag\\
&= I(X;B|Y) - I(A;M|X,Y,B) + I(X,A;M|Y,B) \notag\\
&\ge I(X;B|Y) - H(A|X,Y,B) + I(X,A;M|Y,B) \notag\\
&= H(B|Y) - H(A,B|X,Y) + I(U_A;M|U_B),\label{eq:ic12}
\end{align}
where (a) follows from the Markov chain $B - (Y,M) - X$.
Similarly,
\begin{align}
\PiIC[(12)]{XY} \ge H(A|X) - H(A,B|X,Y) + I(U_B;M|U_A).\label{eq:ic21}
\end{align}
Then it is enough to outer bound the region containing
$(I(U_A;M|U_B),I(U_B;M|U_A))$. Let $V_A=(U_A,M)$, $V_B=(U_B,M)$, the views of Alice and Bob at the end of
the protocol. By~\Propositionref{monotonicity},
\begin{align*}
\Rtens{V_B}{V_A} \supseteq \Rtens{Y}{X}.
\end{align*}
This implies that, for each $\prob_{Q|X,Y}$, there exists a
$\prob_{\tilde{Q}|V_A,V_B}$ such that,
\begin{align}
I(V_A;\tilde{Q}|V_B) &\leq I(X;Q|Y),\label{eq:proof1}\\
I(V_B;\tilde{Q}|V_A) &\leq I(Y;Q|X),\\
I(V_A;V_B|\tilde{Q}) &\leq I(X;Y|Q).\label{eq:proof3}
\end{align}
But,
\begin{align*}
I(V_A;\tilde{Q}|V_B) &= I(U_A,M;\tilde{Q}|U_B,M)= I(U_A;\tilde{Q},M|U_B) - I(U_A;M|U_B)\\ &\geq I(U_A;\tilde{Q}|U_B) - I(U_A;M|U_B)
\end{align*}
Similarly,
\begin{align*}
I(V_B;\tilde{Q}|V_A) &\geq I(U_B;\tilde{Q}|U_A) - I(U_B;M|U_A),\\
I(V_B;V_A|\tilde{Q}) &\geq I(U_A;U_B|\tilde{Q}).
\end{align*}
Using these in \eqref{eq:proof1}-\eqref{eq:proof3}, we have that for all
$Q$, there exists $\tilde{Q}$ such that
\begin{align*}
I(U_A;M|U_B) + I(X;Q|Y) &\geq I(U_A;\tilde{Q}|U_B),\\ 
I(U_B;M|U_A) + I(Y;Q|X) &\geq I(U_B;\tilde{Q}|U_A),\\ 
  I(X;Y|Q) &\geq I(U_A;U_B|\tilde{Q}).\\ 
\end{align*}
Hence, for every $c\ge 0$, we have 
\[(I(U_A;M|U_B),I(U_B;M|U_A))+ \RT_c(Y;X) \subseteq \RT_c(U_B;U_A) .\]
In other words, $(I(U_A;M|U_B),I(U_B;M|U_A))$ must lie in the set
$\RT_c(U_B;U_A) - \RT_c(Y;X)$. Combined with \eqref{eq:ic12} and
\eqref{eq:ic21}, we get that
$(\PiIC[(12)]{XY},\PiIC[(21)]{XY})$ $\in$ $(H(B|Y)-H(AB|XY),H(A|X)-H(AB|XY)) + $ $\bigcap_{c\geq0} \RT_c(U_B;U_A) - \RT_c(Y;X)$.
Since this holds for all $\Pi$ such that $\Pi(X;Y)\outputs(A;B)$, we get
\begin{align*}
\ICost{X;Y}{A;B} \subseteq (H(B|Y)-H(AB|XY),H(A|X)-H(AB|XY)) + \Outer{X;Y}{A;B} = \Outeradj{X;Y}{A;B}.
\end{align*}
\end{proof}

\begin{corol}
\label{cor:tens-ic}
For all $X,Y,Z$, \[ \IC{XY}{Z} \ge \Twyn(XZ;YZ) - \Twyn(X;Y). \]
In particular, if $X$ and $Y$ are independent of each other, 
$\IC{XY}Z \ge \Twyn(XZ;YZ)$.
\end{corol}
\begin{proof}
Firstly, note that the condition in \Theoremref{tens-ic} holds when $A=B=Z$,
since $H(Z|X)\le H(Z|XY)$ and $H(Z|Y)\le H(Z|XY)$. Thus,
\[ \ICost{X;Y}{Z;Z} \subseteq  \Outer{X;Y}{Z;Z} \subseteq \RT_0(YZ;XZ) - \RT_0(Y;X).\]
Then, $\IC{XY}Z = \inf_{(a,b)\in\ICost{X;Y}{Z;Z}} (a+b) \ge \inf_{(a,b)\in \RT_0(YZ;XZ) - \RT_0(Y;X)} (a+b)$.
Now,  $\forall (a,b)\in(S_1-S_2)$, we have $S_1 \supseteq (a,b)+S_2$; hence,
\[ \inf_{(r_1,r_2)\in S_1} (r_1+r_2) \le \inf_{(a,b)\in S_1-S_2} (a+b)+\inf_{(r_1,r_2)\in S_2} (r_1+r_2).\]
Recall that $\inf_{(r_1,r_2)\in \RT_0{U;V}} (r_1+r_2) = \Twyn(U;V)$. Thus,
\[ \IC{XY}Z \ge \Twyn(YZ;XZ)-\Twyn(Y;X).\]
The statement in the theorem follows from the symmetry of \Twyn.
\end{proof}

\section{Information Complexity vs.\ Communication Complexity}
\label{sec:ic-cc}
Below we show that the communication complexity region is outer-bounded by
the information complexity region. We start with \Lemmaref{ic-cc} below,
which relates the communication cost pair of a protocol to its information cost
pair. A simplified version of this result that
has been used extensively, namely, $\PiIC{XY}  \le \Piratex$, appears in
\cite{BravermanRa11}. Note that from  \Lemmaref{ic-cc} it follows that, in
fact, $\PiIC{XY} \le \Pirate{XY}$ (and clearly, $\Pirate{XY} \le \Piratex$).
That is, the information-complexity lower-bound applies not just to the
worst case communication complexity, but also to the average case
communication complexity.
\begin{lemma}
\label{lem:ic-cc}
For any protocol $\Pi$ and input distribution $(X,Y)$, the following hold:
\begin{align*}
\PiIC[(12)]{XY} \le  \Pirate[(12)]{XY}, &&
\PiIC[(21)]{XY} \le  \Pirate[(21)]{XY}.
\end{align*}
In particular, $\PiIC{XY} \le \Pirate{XY}$.
\end{lemma}
\begin{proof}
We shall show that 
$\PiIC[(12)]{XY} \le  \Pirate[(12)]{XY}$; the second inequality follows
similarly, and the third is obtained by adding the first two inequalities.
Below, the random variable $M$ denotes the transcript of the protocol $\Pi$
with input $(X;Y)$, $M_i$ denotes the $i^\text{th}$ bit of $M$, and
$M^i$ denotes the first $i$ bits of $M$. For notational convenience, we
define $M_i$ to be a fixed symbol (say, 0) if $i$ is greater than the length
of $M$. Let \cM be the set of all complete transcripts.%
\footnote{Since we do not require the transcripts to be parsable on their
own without an input (see \Footnoteref{model}), strictly speaking, the set
of complete transcripts is not well-defined. However, \cM can be defined
more loosely as, for instance, the set of all strings of length $d$, where
$d$ is an upperbound on the worst-case communication cost of the protocol,
and the arguments in the proof continue to hold. In fact, even if this cost is
unbounded, but as long as the average cost \Pirate[(12)]{XY} is bounded
(otherwise the inequality is trivial to see), it is possible to extend the
proof by considering $d\rightarrow\infty$.}
Also, for $m\in\cM$, we write $|m|_{12}$ to denote the (expected) number of
bits in $m$ that are sent by party 1 to party 2 (expectation over either
input), and similarly $|m|_{21}$ to denote the bits in the other direction,
so that $|m|=|m|_{12}+|m|_{21}$.

\begin{align*}
\PiIC[(12)]{XY} 
&= I(M; X | Y) = \sum_{i=0}^{\infty} I(M_{i+1}; X | Y, M^i) \\
&= \sum_{i=0}^{\infty} \sum_{m\in\{0,1\}^i} \pr[M^i=m] \cdot I(M_{i+1}; X | Y, M^i=m) \\
&= \sum_{i=0}^{\infty} \sum_{m\in\{0,1\}^i} 
		\left( \sum_{\substack{\widehat{m}\in\cM:\\m=\widehat{m}^i}} \pr[M=\widehat m] \right)
		\cdot I(M_{i+1}; X | Y, M^i=m) \\
&= \sum_{i=0}^{\infty} \sum_{\widehat{m}\in\cM}  \pr[M=\widehat m] I(M_{i+1}; X | Y, M^i=\widehat{m}^i) \\
&= \sum_{\widehat{m}\in\cM} \pr[M=\widehat m] \cdot \sum_{i=0}^{|\widehat{m}|-1}  I(M_{i+1}; X | Y, M^i=\widehat{m}^i) \\
&\stackrel{\text{(a)}}{\le} 
	\sum_{\widehat{m}\in\cM} \pr[M=\widehat m] \cdot |\widehat{m}|_{12} = \Pirate[(12)]{XY}
\end{align*}
where inequality (a) follows from the fact that, for each value of $y$, $I(M_{i+1}; X | Y=y, M^i=\widehat{m}^i) = 0$ if, after
$\widehat{m}^i$ (and given $Y=y$), the next message is sent by Bob, and otherwise $I(M_{i+1}; X | Y=y, M^i=\widehat{m}^i) \le H(M_{i+1}) \le 1$.
\end{proof}

The following theorem is an immediate consequence of \Lemmaref{ic-cc}.
\begin{theorem}
\label{thm:ic-cc} $\Capacity{X;Y}{A;B} \subseteq \ICost{X;Y}{A;B}$.
\end{theorem}
\begin{proof}
Consider any protocol $\Pi$ that takes $(X;Y)$ as input and outputs $(A;B)$.
By \Lemmaref{ic-cc}, $\PiIC[(12)]{XY} \le  \Pirate[(12)]{XY}$ and
$\PiIC[(21)]{XY} \le  \Pirate[(21)]{XY}$. Thus, by definition of
\ICost{X;Y}{A;B}, $(\Pirate[(12)]{XY},\Pirate[(21)]{XY}) \in \ICost{X;Y}{A;B}$.
Since this holds for all $\Pi$ such that $\Pi(X;Y)\outputs(A;B)$,
and \ICost{X;Y}{A;B} is an upward closed region, the theorem follows.
\end{proof}

Following the definitions, the above theorem yields the following
lower-bound:
\[\CC[\eps]{X;Y}f \ge \inf_{\substack{\prob_{Z|XY}:\\\SDiff(\prob_{ZXY},\prob_{f(X,Y)XY})\le\eps}} \IC{XY}Z.\]
Combining this with \Corollaryref{tens-ic}, we obtain the following
lower-bound on (average-case) communication complexity.
\begin{corol}
\label{cor:tens-cc}
For all $\eps\ge0$,
$\displaystyle\CC[\eps]{X;Y}f \ge \inf_{\substack{\prob_{Z|XY}:\\\SDiff(\prob_{ZXY},\prob_{f(X,Y)XY})\le\eps}} \Twyn(XZ;YZ) - \Twyn(X;Y).$

In particular, if $(X,Y)$ are independent of each other,
$\displaystyle\CC[\eps]{X;Y}f \ge \inf_{\substack{\prob_{Z|XY}:\\\SDiff(\prob_{ZXY},\prob_{f(X,Y)XY})\le\eps}} \Twyn(XZ;YZ).$
\end{corol}

\section{Bounding Tension Using Discrepancy}
\label{sec:disc-tens}

\begin{theorem}
\label{thm:disc-tens}
Suppose $(X,Y)$ are independent random variables over $\cX\times\cY$,
and $f:\cX\times\cY\rightarrow\{0,1\}$ is a function with
$\Discrepancy{f}{XY} \le \disc$. Also, suppose $Z$ is a binary random
variable jointly distributed with $(X,Y)$ such that 
$\pr[Z\not=f(X,Y)] \le \frac12 - \eps$. Then
\[ \Twyn(XZ;YZ) \ge \frac\eps{2(1-\eps)} \log \frac{\eps}{\disc} - 4.\]
\end{theorem}

\begin{proof}
We seek to lower-bound the tension, $\Twyn(XZ;YZ) = \inf_{Q: XZ-Q-YZ}
I(XZ;Q|YZ)+I(YZ;Q|XZ)$.  Consider a random variable $Q$ over an alphabet
\cQ, jointly distributed with $(X,Y)$, such that $XZ-Q-YZ$. Firstly, note
that this implies $H(Z|Q)=0$, and $I(X;Y|Q)=0$ (since both these quantities
are upper-bounded by $I(XZ;YZ|Q) = 0$).  To lower-bound
$I(XZ;Q|YZ)+I(YZ;Q|XZ)$, it is enough to lower-bound
$I(XY;Q)$, as shown below:
\begin{align*}
I(XZ;Q|YZ)+I(YZ;Q|XZ) 
&= I(X;Q|YZ) + I(Y;Q|XZ)\\
&= I(XZ;Q|Y) - I(Z;Q|Y) + I(YZ;Q|X) - I(Z;Q|X)\\
&\ge I(X;Q|Y) - 1 +I(Y;Q|X) - 1\\
&= (I(X;Q|Y)+ I(Y;Q|X) + I(X;Y)) - I(X;Y) - 2\\
&= (I(XY;Q) + I(X;Y|Q)) - I(X;Y) - 2\\
&= I(XY;Q) - I(X;Y)-2,
\end{align*}
where in the last step we used the fact that $I(X;Y|Q)=0$. Since we are given that $X$ and $Y$ are independent, we have
$I(XZ;Q|YZ)+I(YZ;Q|XZ) \ge I(XY;Q)-2$.

For all $q\in\cQ$, let $\Adv{q} = | \pr[f(X,Y)=0|Q=q] - \pr[f(X,Y)=1|Q=q]|$.
\begin{align*}
2\eps 
&\le \pr[Z=f(X,Y)]-\pr[Z\not=f(X,Y)] \\
&= \sum_{q\in\cQ} \pr[Q=q] \left(\pr[Z=f(X,Y)|Q=q]-\pr[Z\not=f(X,Y)|Q=q]\right) \\
&\le  \sum_{q\in\cQ} \pr[Q=q] \Adv{q},
\end{align*}
where in the last step we used the fact that $H(Z|Q)=0$.

We shall define an auxiliary random variable $R$ over all rectangles (i.e.,
with alphabet $\cR=\{ \cX'\times\cY' : \cX'\subseteq \cX, \cY' \subseteq \cY \}$),
jointly distributed with $(X,Y,Q)$, satisfying that the following conditions
for each $q\in\cQ$. Below, let $\cR_0 \subseteq \cR$ denote the set of ``small'' rectangles: i.e.,
$\cR_0 = \{ r \in \cR : \pr[(X,Y)\in r]< \alpha \}$, where $\alpha$ is a
parameter to be set later.
Also, for 
$q\in\cQ$, let $\cL_q\subseteq\cX\times\cY$ denote the set of all $(x,y)$ which
lie in the small rectangles that occur with $q$; i.e., 
\[ \cL_q = {\displaystyle \bigcup_{\substack{r\in\cR_0:\\\pr[Q=q,R=r]>0}}} r.\]

\begin{claim}
\label{clm:decomposeQ}
There exists a random variable $R$ with alphabet \cR, 
jointly distributed with $(X,Y,Q)$ such that for each
$q\in\cQ$ the following hold.
\begin{itemize}
\item For every $r\in\cR$ such that $\pr[Q=q,R=r]>0$, the distribution
$\prob_{XY|Q=q,R=r}$ is the same as $\prob_{XY|(X,Y)\in r}$ (i.e., $\prob_{XY}$
restricted to the rectangle $r$).
\item  $\pr[(X,Y) \in \cL_q] \le 2\sqrt{\alpha}$.
\end{itemize}
\end{claim}
We prove this claim in \Appendixref{proofs}.

Let $\RR$ be a boolean random variable such that $\RR=0$ iff
$R\in\cR_0$, and $\RR=1$ otherwise. Let $\QQ=(Q,\RR)$. Note that $I(XY;Q)
\ge I(XY;\QQ)-1$; so it is sufficient to lower-bound $I(XY;\QQ)$.

First, we lower-bound $\pr[\RR=0]$, relying on the upper bound on
discrepancy. Let $\Adv{q,r} = | \pr[f(X,Y)=0|Q=q,R=r] - \pr[f(X,Y)=1|Q=q,R=r]|$.
Then $\Adv{q} \le \sum_r \pr[R=r|Q=q]\Adv{q,r}$. Further,
\begin{align*}
\pr[(X,Y)\in r] \cdot \Adv{q,r}
&=
\pr[(X,Y)\in r] \cdot  | \pr[f(X,Y)=0 | (X,Y)\in r] - \pr[f(X,Y)=1 | (X,Y)\in r]| 
	\\&\qquad\qquad\text{\hfill since } \prob_{XY|Q=q,R=r}\equiv \prob_{XY|(X,Y)\in r}\\
&= |\pr[f(X,Y)=0 \land (X,Y)\in r] - \pr[f(X,Y)=1 \land (X,Y)\in r]| \\
&\le \Discrepancy{f}{XY} \le \disc.
\end{align*}
Then,
since $\pr[(X,Y)\in r]\ge\alpha$ for $r\not\in\cR_0$, we conclude
that $\Adv{q,r} \le \frac\disc\alpha$, for $r\not\in\cR_0$. 
Now,
\begin{align*}
2\eps 
&\le \sum_{q\in\cQ} \pr[Q=q]\Adv{q} \le \sum_{q,r\in\cR} \pr[Q=q,R=r]\Adv{q,r} \\
&\le \sum_{q,r\in\cR_0}\pr[Q=q,R=r] + \sum_{q,r\not\in\cR_0} \pr[Q=q,R=r] \Adv{q,r} \\
&\le \sum_{q,r\in\cR_0}\pr[Q=q,R=r] + \frac\disc\alpha \sum_{q,r\not\in\cR_0} \pr[Q=q,R=r] \\
&\le \pr[\RR=0] + \frac\disc\alpha(1-\pr[\RR=0]).
\end{align*}
So, $\pr[\RR=0]\ge\frac{2\eps-\frac\disc\alpha}{1-\frac\disc\alpha}$.

Finally, we use the following lemma, proven in \Appendixref{proofs} (with
$S=(X,Y)$, $T=\QQ$ and $\cT_0 = \cQ\times\{0\}$) to obtain our lower bound
on $I(XY;\QQ)$.
\begin{lemma}
\label{lem:Ibound}
Let $S,T$ be jointly distributed random variables over $\cS\times\cT$, 
and $\cT_0\subseteq\cT$ be such that $\forall t\in\cT_0$, $\pr[S\in\cS_t] \le
\delta$ where $\cS_t = \{ s\in\cS : \pr[S=s|T=t]>0 \}$, and $\pr[T\in\cT_0] \ge \varepsilon$.
Then, $I(S;T) \ge \varepsilon \log \frac1\delta$.
\end{lemma}
We apply this lemma with $\delta=2\sqrt{\alpha}$ and
$\varepsilon=\frac{2\eps-\frac\disc\alpha}{1-\frac\disc\alpha}$. This yields 
$I(XY;\QQ) \ge\frac{2\eps-\frac\disc\alpha}{1-\frac\disc\alpha}(\frac12 \log \frac1\alpha - 1)$.
As described above, this bound on $I(XY;\QQ)$ yields the following bound on tension:
\begin{align}
\label{eq:discr-bound}
\Twyn(XZ;YZ) &\ge 
\frac{2\eps-\frac\disc\alpha}{1-\frac\disc\alpha}(\frac12 \log \frac1\alpha - 1) - 3.
\end{align}
To complete the proof, we set $\alpha=\frac\disc\eps$, and note that since
$\eps<\frac12$, we have $\eps/(1-\eps) < 1$.
\end{proof}

\paragraph{Remark:}
Often \disc is a quantity that vanishes as a size parameter of the
inputs grows (e.g., when $f$ is the inner-product function). When
$\eps\cdot\log\frac\eps\disc=\omega(1)$, one can obtain a tighter bound from
the above proof, by setting $\alpha=\left(\frac{\disc}\eps\right)^{1-\beta}$
for a small enough $\beta>0$.  This gives $\Twyn(XZ;YZ) \ge
\eps\cdot\log\frac\eps\disc\cdot(1-o(1))$.

\section*{Acknowledgments}
We gratefully acknowledge Mark Braverman, Prahladh Harsha and Rahul Jain for
helpful discussions and pointers.
\bibliographystyle{alpha}
\bibliography{bib}

\appendix
\section{On The Nature of Tension Region}
\label{app:tension}

In this appendix we present a gentle introduction to the notion of 
tension region, as developed in~\cite{PrabhakaranPr14}. We refer the
interested readers to \cite{PrabhakaranPr14} for more details.

Consider the random variables $X=(X',Q)$ and $Y=(Y',Q)$ where $X',Y',Q$ are
independent. In this case, it is natural to consider $Q$ as the common random variable of
$X$ and $Y$ and $H(Q)$ as a natural measure of ``common information.'' $Q$ is
determined both by $X$ and by $Y$ individually. Moreover, conditioned on
$Q$, $X$ and $Y$ are independent, i.e., $X-Q-Y$ is a Markov chain.  One
could extend this to arbitrary $X,Y$, in a couple of natural ways. The
approach of G\'{a}cs and K\"{o}rner~\cite{GacsKo73}%
is to find the ``largest'' random variable $Q$ (largness being measured in
terms of entropy) such that it is determined by $X$ alone as well as by $Y$
alone (with probability 1):
\begin{align*}
\CIgk(X;Y) &= \max_{\substack{\prob_{Q|XY}:\\H(Q|X)=H(Q|Y)=0}} H(Q)\notag\\
&= I(X;Y) - \min_{\substack{\prob_{Q|XY}:\\H(Q|X)=H(Q|Y)=0}}
I(X;Y|Q).%\label{eq:introGKCI}
\end{align*}
Clearly $\CIgk(X;Y)\leq I(X;Y)$ and, in general, this inequality maybe strict, i.e., common
information, in general, does not account for all the dependence between
$X$ and $Y$.

Wyner gave a different generalization~\cite{Wyner75}
where he defined common information in terms of the ``smallest'' random variable $Q$ (smallness 
being measured in terms of $I(XY;Q)$) so that $X$ and $Y$ are independent conditioned on $Q$ .
\begin{align*}
\CIwyn(X;Y) &= \min_{\substack{\prob_{Q|XY}:\\X-Q-Y}} I(XY;Q) \notag\\
&= I(X;Y) + \min_{\substack{\prob_{Q|XY}:\\X-Q-Y}} (I(Y;Q|X)+I(X;Q|Y)). %\label{eq:introWynCI}
\end{align*}
Now $\CIwyn(X;Y)\geq I(X;Y)$.
When $X,Y$
are of the form $X=(X',Q)$ and $Y=(Y',Q)$, where $X',Y',Q$ are independent,
then there indeed is a unique interpretation of common information (when
$\CIgk(X;Y)=\CIwyn(X;Y)=H(Q)$). Between these
extremes represented by these two measures, there are several ways
in which one could define a random variable to capture the dependence 
between $X$ and $Y$.

\begin{definition}
For a pair of correlated random variables $(X,Y)$, and
$\prob_{Q|XY}$, we say $Q$ {\em perfectly resolves} $(X,Y)$ if $I(X;Y|Q)=0$ and
$H(Q|X)=H(Q|Y)=0$. We say $(X,Y)$ is {\em perfectly resolvable}
if there exists $\prob_{Q|XY}$ such that $Q$ perfectly resolves $(X,Y)$.
\end{definition}
If $(X,Y)$ is perfectly resolvable, then $\CIgk(X;Y)=I(X;Y)=\CIwyn(X;Y)$
represents the entire mutual information between them. Tension region \Rtens XY can be thought of as measuring the extent to which a pair of random variables $(X,Y)$ is {\em not} resolvable.

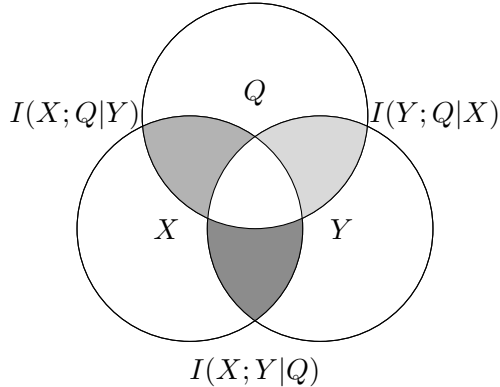
\begin{figure}[htb]
\centering

\def\firstcircle{(210:1cm) circle (1.5cm)}
\def\secondcircle{(90:1cm) circle (1.5cm)}
\def\thirdcircle{(330:1cm) circle (1.5cm)}

\begin{tikzpicture}
    \draw \firstcircle node[left] {$X$};
    \draw \secondcircle node [above] {$Q$};
    \draw \thirdcircle node [right] {$Y$};

    \begin{scope}
      \clip \firstcircle;
      \fill[gray!60] \secondcircle;
    \end{scope}
    \begin{scope}
      \clip \thirdcircle;
      \fill[gray!30] \secondcircle;
    \end{scope}
    \begin{scope}
      \clip \firstcircle;
      \fill[gray!90] \thirdcircle;
    \end{scope}

    \begin{scope}
      \clip \firstcircle;
      \clip \secondcircle;
      \fill[white] \thirdcircle;
    \end{scope}

\node at (23:2.6cm) {$I(Y;Q|X)$};
\node at (157:2.6cm) {$I(X;Q|Y)$};
\node at (270:2.4cm) {$I(X;Y|Q)$};

    \draw \firstcircle; 
    \draw \secondcircle; 
    \draw \thirdcircle;

\end{tikzpicture}

\caption{A Venn diagram representation of the three coordinates $\big(I(Y;Q|X),I(X;Q|Y),I(X;Y|Q)\big)$ in the definition of $\Rtens XYQ$. Figure taken from~\cite{PrabhakaranPr14}.}
\label{fig:venn}
\end{figure}
Recall the definition of {\em tension region} \Rtens AB of a pair of random variables $A,B$:
\begin{align*}
\Rtens AB = \{ (r_1,r_2,r_3):&\; \exists Q \text{ jointly distributed with } A,B \\
          &\text{ s.t. } I(B;Q|A) \le r_1, I(A;Q|B) \le r_2, I(A;B|Q) \le r_3 \}.
\end{align*}
It follows from Fenchel-Eggleston's strengthening of Carath\'{e}odory's
theorem~\cite[pg. 310]{CsiszarKo81}, that we can restrict ourselves to
$\prob_{Q|XY}$ with alphabet \cQ such that $|\cQ|\leq |\cX||\cY|+2$.

\begin{figure}[htb]
\centering
\scalebox{0.6}{\includegraphics{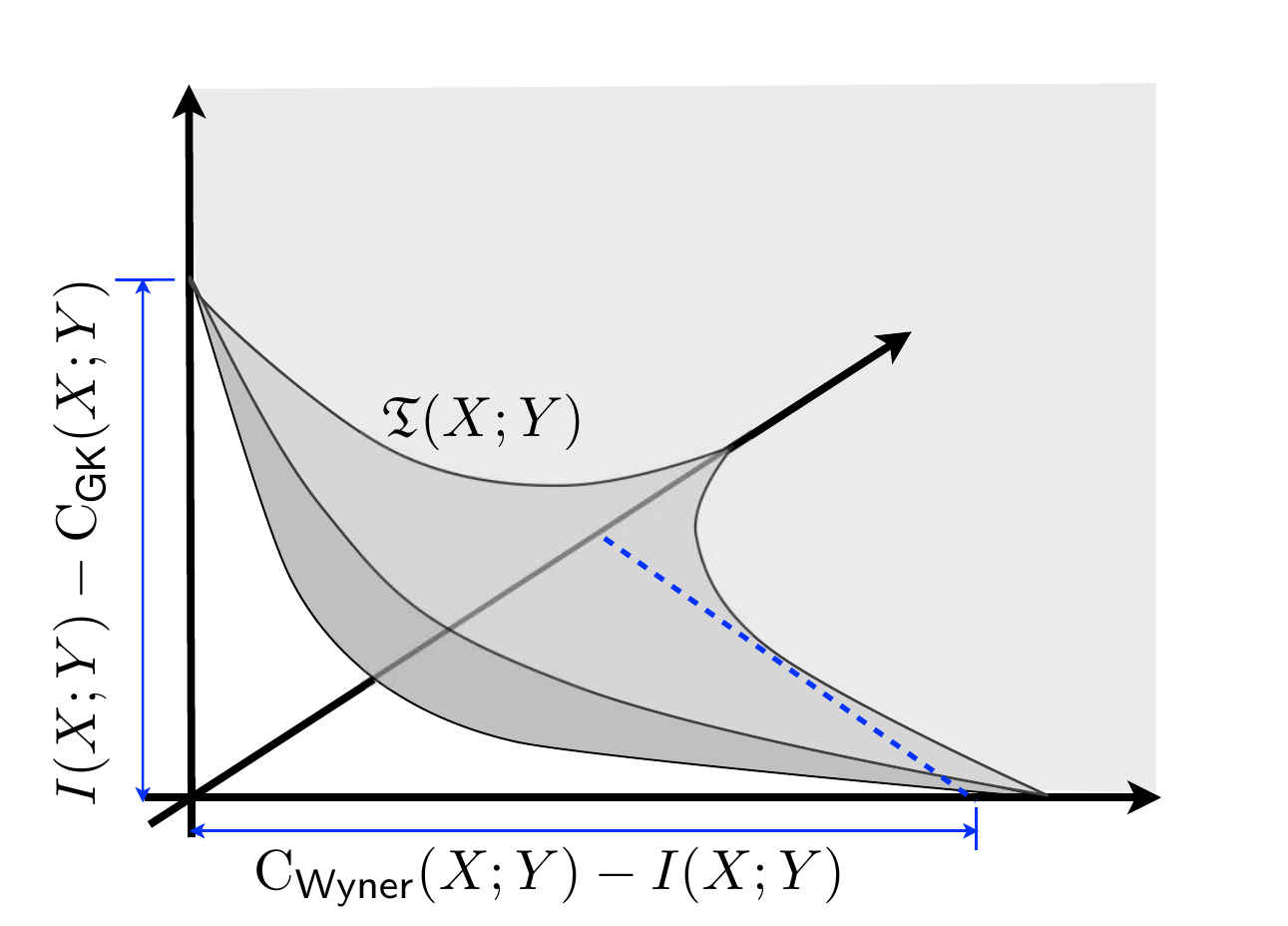}}
\caption{A schematic representation of the region \Rtens XY.
\Rtens XY is an unbounded, convex region, bounded away from the origin
(unless $(X,Y)$ is perfectly resolvable). Relationship between two
points on the boundary of \Rtens XY and the quantities
\CIgk(X;Y) and \CIwyn(X;Y) 
(The dotted line is at 45$^\circ$
to the axes.)  Figure taken from~\cite{PrabhakaranPr14}.}
\label{fig:Rtens}
\end{figure}
It can be shown that \Rtens XY includes the origin 
if and only if the pair $(X,Y)$
is perfectly resolvable. When this is not the case, it is important to consider all three coordinates of
 together to identify the unresolvable nature of a pair $(X,Y)$,
because since  \Rtens XY does intersect each of the three
axes, or in other words, any two coordinates of can be made
simultaneously 0 by choosing an appropriate $Q$.

Below we summarize several useful properties of \Rtens XY. For interpretations of \Rtens XY in terms of certain information theoretic problems, we refer the reader to~\cite{PrabhakaranPr14}.

\subsection{Some Properties of Tension}

\paragraph{Monotonicity of \Rtens XY.} Wolf and Wullschleger~\cite{WolfWu05} showed that the three axes incercepts have a certain
``monotonicity'' property (they can only decrease, as $X,Y$ evolve as the
views of two parties in a protocol).  In fact, this
monotinicity  is a consequence of the
monotinicity of the entire region \Rtens XY stated in \Propositionref{monotonicity}.

\paragraph{Tensorization of \Rtens XY.}
If $(X_1,Y_1)$ is independent of $(X_2,Y_2)$,then 
\[\Rtens {(X_1X_2)} {(Y_1Y_2)} = \Rtens {X_1} {Y_1} + \Rtens {X_2} {Y_2}.\]

\paragraph{Convexity, closedness, and continuity of \Rtens XY.}
Firstly, the region of tension is closed and convex.
Secondly, the region of tension is {\em continuous} in the sense
that when the joint p.m.f.
$p_{X,Y}$ is close to the joint p.m.f. $p_{X',Y'}$, the tension
regions $\Rtens {X}{Y}$ and $\Rtens {X'}{Y'}$ are also close. 
Specifically, if $\SDiff(XY,X'Y')\le \epsilon$, then
$\Rtens {X}{Y} \subseteq \Rtens {X'}{Y'} - \delta(\epsilon)$,
where $\delta(\epsilon) = 2H_2(\epsilon) + \epsilon\log\max\{|\cX|,|\cY|\}$.

\section{Proof of \Lemmaref{Ibound} and \Claimref{decomposeQ}.} 
\label{app:proofs}

To
complete the proof of \Theoremref{disc-tens} we need to prove
\Lemmaref{Ibound} and \Claimref{decomposeQ}. We do this below.

\begin{proof}[Proof of \Lemmaref{Ibound}]
We have
\begin{align*}
I(S;T) &= \sum_{(s,t)\in\cS\times\cT} \prob_{S,T}(s,t)\log\frac{\prob_{S,T}(s,t)}{\prob_S(s)\prob_T(t)}\\
 &= \sum_{t\in\cT} \prob_T(t) \sum_{s\in\cS_t} \prob_{S|T}(s|t) \log\frac{\prob_{S|T}(s|t)}{\prob_S(s)}\\
 &= \sum_{t\in\cT_0} \prob_T(t) \sum_{s\in\cS_t} \prob_{S|T}(s|t) \log\frac{\prob_{S|T}(s|t)}{\prob_S(s)} +
\sum_{t\in\cT-\cT_0} \prob_T(t) \sum_{s\in\cS_t} \prob_{S|T}(s|t) \log\frac{\prob_{S|T}(s|t)}{\prob_S(s)}
\end{align*}
Notice that, for each $t$
\[ \sum_{s\in\cS_t} \prob_{S|T}(s|t) \log\frac{\prob_{S|T}(s|t)}{\prob_S(s)} = D(\prob_{S|T=t}\|\prob_S) \geq 0.\]
Hence, we have
\[I(S;T) \geq \sum_{t\in\cT_0} \prob_T(t) \sum_{s\in\cS_t} \prob_{S|T}(s|t) \log\frac{\prob_{S|T}(s|t)}{\prob_S(s)}.\]
For each $t\in\cT_0$, let $p_t=\pr[S\in\cS_t]=\sum_{s\in\cS_t} \prob_S(s)$, and let us define over $\cS_t$ the probability mass function, $\prob_{(t)}(s)=\frac{\prob_S(s)}{p_t}$, $s\in\cS_t$. Note that $p_t\le\delta$. Then, for $t\in \cT_0$,
\begin{align*}
\sum_{s\in\cS_t} \prob_{S|T}(s|t) \log\frac{\prob_{S|T}(s|t)}{\prob_S(s)}
&= \sum_{s\in\cS_t} \prob_{S|T}(s|t) \log\frac{\prob_{S|T}(s|t)}{\prob_S(s)/p_t}\frac{1}{p_t}\\
&= D(\prob_{S|T}||\prob_{(t)}) + \log\frac{1}{p_t}\\
&\ge \log\frac{1}{\delta}.
\end{align*}
Subtituting this back,
\begin{align*}
I(S;T) \geq \sum_{t\in\cT_0} \prob_T(t)\log\frac{1}{\delta} \geq \varepsilon \log\frac{1}{\delta}.
\end{align*}
\end{proof}

\begin{proof}[Proof of \Claimref{decomposeQ}]
It remains to describe the distribution $\prob_{R|XYQ}$ so that the
conditions listed in \Claimref{decomposeQ} hold.

For $r=\cX_r\times\cY_r\in\cR$, we let  
\begin{align*}
\sigma_{q,r} &= \min_{x\in \cX_r}\frac{\pr[X=x,Q=q]}{\pr[X=x]\pr[Q=q]}-\max_{x'\not\in \cX_r} \frac{\pr[X=x',Q=q]}{\pr[X=x']\pr[Q=q]}\\
\tau_{q,r} &= \min_{y\in \cY_r}\frac{\pr[Y=y,Q=q]}{\pr[Y=y]\pr[Q=q]}-\max_{y'\not\in \cY_r} \frac{\pr[Y=y',Q=q]}{\pr[Y=y']\pr[Q=q]}
\end{align*}
Above, in defining $\max_{x'\not\in \cX_r}$, if no such $x'$ exists -- i.e.,
$\cX_r=\cX$ -- we take the maximum to be 0 (and similarly for $\max_{y'\not\in \cY_r}$).
Now we define $\prob_{R|XYQ}$ as follows:
\begin{align*}
\pr[R=r|X=x,Y=y,Q=q] &= 
\begin{cases}
\sigma_{q,r}\cdot\tau_{q,r}\cdot\frac{\pr[X=x,Y=y]}{\pr[X=x,Y=y|Q=q]}
	& \text{ if } \sigma_{q,r}>0, \tau_{q,r}>0 \text{ and } (x,y)\in r \\
0	& \text{ otherwise.}
\end{cases} \\
\end{align*}

An alternate way to describe the mass assigned to $r$ is as follows. Let
$\cX_q\times\cY_q$ be the support of $\prob_{XY|Q=q}$. Let $\cX_q = \{
x_1,\cdots,x_M \}$, such that $\frac{\pr[X=x_i,Q=q]}{\pr[X=x_i]\pr[Q=q]} \ge
\frac{\pr[X=x_{i+1},Q=q]}{\pr[X=x_{i+1}]\pr[Q=q]}$ for all $i\in[1,M-1]$.
For notational convenience, we also define a dummy $x_{M+1}$ with
$\frac{\pr[X=x_{M+1},Q=q]}{\pr[X=x_{M+1}]\pr[Q=q]}=0$. 
Define $y_1,\cdots,y_N,y_{N+1}$ similarly, where $N=|\cY_q|$. 
Then, the only rectangles $r$ for which $\pr[R=r|Q=q]$ can be positive are of the form
$r_{ij} = \cX_i \times \cY_j$ for $(i,j)\in[M]\times[N]$, where $\cX_i = \{x_1,\cdots,x_i\}$,
$\cY_j=\{y_1,\cdots,y_j\}$, 
$\frac{\pr[X=x_i,Q=q]}{\pr[X=x_i]\pr[Q=q]} > \frac{\pr[X=x_{i+1},Q=q]}{\pr[X=x_{i+1}]\pr[Q=q]}$, and
$\frac{\pr[Y=y_j,Q=q]}{\pr[Y=y_j]\pr[Q=q]} > \frac{\pr[Y=y_{j+1},Q=q]}{\pr[Y=y_{j+1}]\pr[Q=q]}$.

First, we verify that $\prob_{R|Q=q,X=x,Y=y}$ is indeed a valid probability distribution.
\begin{align*}
&\sum_{r\in\cR} \pr[R=r|Q=q,X=x_{i^*},Y=y_{i^*}] \\
&\qquad\qquad =	\sum_{r: (x_{i^*},y_{i^*})\in r}
		\sigma_{q,r}\cdot\tau_{q,r}\cdot\frac{\pr[X=x_{i^*},Y=y_{i^*}]}{\pr[X=x_{i^*},Y=y_{i^*}|Q=q]} \\
&\qquad\qquad =	\frac{\pr[X=x_{i^*},Y=y_{i^*}]}{\pr[X=x_{i^*},Y=y_{i^*}|Q=q]} \cdot \sum_{i=i^*}^M \sum_{j=j^*}^N
		\sigma_{q,r_{ij}}\cdot\tau_{q,r_{ij}} \\
&\qquad\qquad =	\frac{\pr[X=x_{i^*},Y=y_{i^*}]}{\pr[X=x_{i^*},Y=y_{i^*}|Q=q]} \cdot 
	\sum_{i=i^*}^M 
	\left(\frac{\pr[X=x_i,Q=q]}{\pr[X=x_i]\pr[Q=q]} - \frac{\pr[X=x_{i+1},Q=q]}{\pr[X=x_{i+1}]\pr[Q=q]}\right)
		\\ &\qquad\qquad\qquad\qquad\cdot 
	\sum_{j=j^*}^N
	\left(\frac{\pr[Y=y_j,Q=q]}{\pr[Y=y_j]\pr[Q=q]} - \frac{\pr[Y=y_{j+1},Q=q]}{\pr[Y=y_{j+1}]\pr[Q=q]}\right) \\
&\qquad\qquad =	\frac{\pr[X=x_{i^*},Y=y_{i^*}]}{\pr[X=x_{i^*},Y=y_{i^*}|Q=q]} \cdot 
		\frac{\pr[X=x_{i^*},Q=q]}{\pr[X=x_{i^*}]\pr[Q=q]} \cdot
		\frac{\pr[Y=y_{j^*},Q=q]}{\pr[Y=y_{j^*}]\pr[Q=q]} = 1,
\end{align*}
where in the last step we used the facts that $X,Y$ are independent and also they are conditionally indepdendent conditioned on $Q$. 

Next, we verify that $\prob_{XY|Q=q,R=r} \equiv \prob_{XY|(X,Y)\in r}$. 
Firstly, if $(x,y)\not\in r$, then $\pr[R=r|X=x,Y=y,Q=q]=0$, and hence
$\pr[X=x,Y=y|Q=q,R=r] = 0$ (and also, $\pr[X=x,Y=y|(X,Y)\in r]=0$). Now,
suppose $(x,y)\in r$. Then,
\begin{align*}
\pr[X=x,Y=y|Q=q,R=r] 
&= \frac{ \pr[R=r|X=x,Y=y,Q=q] \pr[X=x,Y=y|Q=q] }{ \pr[R=r|Q=q] } \\
&= \frac{\sigma_{q,r}\cdot\tau_{q,r}\cdot\pr[X=x,Y=y]}{\pr[R=r|Q=q]} = \frac{\pr[X=x,Y=y]}{F(q,r)},
\end{align*}
where $F(q,r)$ is a quantity independent of $(x,y)$. Since
$\pr[X=x,Y=y|Q=q,R=r]$ is a probability distribution, $F(q,r)=
\sum_{(x,y)\in r} \pr[X=x,Y=y] = \pr[(X,Y)\in r]$.
Thus indeed, 
$\pr[X=x,Y=y|Q=q,R=r] = \pr[X=x,Y=y|(X,Y)\in r]$.

Finally, we argue that $\pr[(X,Y) \in \cL_q] \le 2\sqrt{\alpha}$.
Consider any $q\in\cQ$, and as before,
let $\cX_q = \{x_1,\cdots,x_M\}$, $\cY_q=\{y_1,\cdots,y_N\}$ sorted
appropriately, and, for $i\in[M],j\in [N]$, $r_{ij} =
\{x_1,\cdots,x_i\}\times\{y_1,\cdots,y_j\}$.
Then $(x,y)\in\cL_q$ iff $(x,y)\in r_{ij}$ for some 
$r_{ij}\in\cR_0$ (i.e., $\pr[(X,Y)\in r_{ij}]\le\alpha$).  
Let $i^*$ be the maximum value in $[M]$ such that $\pr[X\in\{x_1,\cdots,x_{i^*}\}]\le\sqrt\alpha$, and
similarly, let $j^*$ be the maximum value in $[N]$ such that $\pr[Y\in\{y_1,\cdots,y_{j^*}\}]\le\sqrt\alpha$. 
Then we note that, if $i>i^*$ and $j>j^*$, then $(x_i,y_j)\not\in \cL_q$.
This is because,  $(x_i,y_j)\in r_{i'j'}$ $\implies$ $(i'\ge i >i^*,j'\ge j >j^*)$ $\implies$ $r_{i'j'} \not\in \cR_0$, as
$\pr[(X,Y)\in r_{i'j'}] =
\pr[X\in\{x_1,\cdots,x_{i'}\}]\cdot\pr[Y\in\{y_1,\cdots,y_{j'}\}] > \sqrt{\alpha}\sqrt{\alpha}$ (by
definition of $i^*$ and $j^*$). Hence,
\begin{align*}
\pr[(X,Y)\in\cL_q]
&\le\pr[(X\in\{x_1,\cdots,x_{i*}\} \lor (Y\in\{y_1,\cdots,y_{j*}\}]\\
&\le\pr[X\in\{x_1,\cdots,x_{i*}\}]+\pr[Y\in\{y_1,\cdots,y_{j*}\}] \le 2\sqrt\alpha.
\end{align*}
\end{proof}

\end{document}